\documentclass[conference]{IEEEtran}
\usepackage{amsmath,graphicx}
\usepackage{amsmath,amsfonts,amssymb}
\usepackage{amsthm}
\IEEEoverridecommandlockouts

\def\qk{q^\mathcal{\omega}}
\title{Classification with High-Dimensional\\ Sparse Samples}
\author{\IEEEauthorblockN{Dayu Huang}
\IEEEauthorblockA{Department of Electrical and Computer Engineering\\
University of Illinois at Urbana-Champaign\\
Urbana, Illinois 61801\\
Email: dhuang8@illinois.edu}
\and
\IEEEauthorblockN{Sean Meyn}
\IEEEauthorblockA{Department of Electrical and Computer Engineering\\University of Florida\\
Gainesville,  FL 32611\\
Email: meyn@ece.ufl.edu} \thanks{\footnotesize The authors would like to acknowledge help discussions with Tu\v{g}kan Batu and Aaron Wagner. 

Financial support from the National Science Foundation (NSF CCF
07-29031 and CCF 08-30776), ITMANET DARPA RK 2006-07284 and AFOSR grant FA9550-09-1-0190
is gratefully acknowledged.  Any opinions, findings,
and conclusions or recommendations expressed in this material are
those of the authors and do not necessarily reflect the views of
NSF, DARPA or AFOSR. }}

 \def\epsy{\varepsilon}

\def\Prob{{\sf P}}

\def\Expect{{\sf E}}

\def\limsup{\mathop{\rm lim{\,}sup}}

\def\sq{\hbox{\rlap{$\sqcap$}$\sqcup$}}
\def\qed{\ifmmode\sq\else{\unskip\nobreak\hfil
\penalty50\hskip1em\null\nobreak\hfil\sq
\parfillskip=0pt\finalhyphendemerits=0\endgraf}\fi\medskip}

 \def\FRAC#1#2#3{\genfrac{}{}{}{#1}{#2}{#3}}

\def\ddtp{{\mathchoice{\FRAC{1}{d^{\hbox to 2pt{\rm\tiny +\hss}}}{dt}}%
{\FRAC{1}{d^{\hbox to 2pt{\rm\tiny +\hss}}}{dt}}%
{\FRAC{3}{d^{\hbox to 2pt{\rm\tiny +\hss}}}{dt}}%
{\FRAC{3}{d^{\hbox to 2pt{\rm\tiny +\hss}}}{dt}}}}

\def\half{{\mathchoice{\FRAC{1}{1}{2}}%
{\FRAC{1}{1}{2}}%
{\FRAC{3}{1}{2}}%
{\FRAC{3}{1}{2}}}}

\def\eqdef{\mathbin{:=}}

\newtheorem{theorem}{Theorem}[section]

\newtheorem{proposition}[theorem]{Proposition}
\newtheorem{lemma}[theorem]{Lemma}

\newtheorem{assumption}[theorem]{Assumption}

\def\Lemma#1{Lemma~\ref{t:#1}}
\def\Proposition#1{Proposition~\ref{t:#1}}
\def\Theorem#1{Theorem~\ref{t:#1}}

\def\Assumption#1{Assumption~\ref{as:#1}}

\def\eq#1/{(\ref{e:#1})}

\newcommand{\field}[1]{\mathbb{#1}}

\def\ind{\field{I}}

\newcounter{rmnum}

\newcounter{anum}

\def\til={{\widetilde =}}

\def\clP{{\cal P}}

\def\bfmx{{\mbox{\protect\boldmath$x$}}}

\def\bfmz{{\mbox{\protect\boldmath$z$}}}
\def\bfmy{{\mbox{\protect\boldmath$y$}}}

\def\bfmX{{\mbox{\protect\boldmath$X$}}}
\def\bfmY{{\mbox{\protect\boldmath$Y$}}}

\def\bfmZ{{\mbox{\protect\boldmath$Z$}}}

\newlength{\noteWidth}
 \long\def\notes#1{\ifinner
             {\tiny #1}
             \else
             \marginpar{\parbox[t]{\noteWidth}{\raggedright\tiny #1}}
             \fi}

\def\notes#1{}

\def\Ho{H0}
\def\Ha{H1}

\def\mP{\Pi_m}
\def\bmP{\clP([m])}
\def\cJ{I}

\def\mub{\eta}
\def\bphi{\mathbf{\phi}}

\newcounter{assumption}
\newtheorem{assumption0}[assumption]{Assumption}

\begin{document}
%\ninept
%
\maketitle
\begin{abstract}
The task of the binary classification problem is to determine which of two distributions has generated a length-$n$ test sequence. The two distributions are unknown; two training sequences of length $N$, one from each distribution, are observed. The distributions share an alphabet of size $m$, which is significantly larger than $n$ and $N$. How does $N,n,m$ affect the probability of classification error? We characterize the achievable error rate in a high-dimensional setting in which $N,n,m$ all tend to infinity, under the assumption that probability of any symbol is $O(m^{-1})$. The results are:
\begin{enumerate}
\item  There exists an asymptotically consistent classifier if and only if $m=o(\min\{N^2,Nn\})$. This extends the previous consistency result in \cite{keltulwagvis10p1478} to the case $N\neq n$.
\item  For the sparse sample case where $\max\{n,N\}=o(m)$, finer results are obtained: The best achievable probability of error decays as $-\log(P_e)=J \min\{N^2, Nn\}(1+o(1))/m$ with $J>0$.
\item A weighted coincidence-based classifier has non-zero generalized error exponent $J$. 
\item The $\ell_2$-norm based classifier has $J=0$.
\end{enumerate}
\end{abstract}

\begin{keywords}
high-dimensional model, large deviations, classification, sparse sample, generalized error exponent \end{keywords}

\section{INTRODUCTION}

Consider the following binary classification problem: Two training sequences $\bfmX\!=\!\{\!X_1,\ldots, X_N\}$ and $\bfmY\!=\!\{\!Y_1, \ldots, Y_N\}$ generated from two different \emph{unknown} sources are observed. The two sources share the same alphabet $[m]\eqdef \{1, \ldots, m\}$. Given a test sequence $\bfmZ=\{Z_1, \ldots, Z_n\}$, the classifier  decides whether $\bfmZ$ comes from the first source or the second. 

The performance of a classifier is usually assessed by how its probability of classification error depends on $N,n,m$. Since the exactly formula for the probability of error is usually complicated, asymptotic models and performance criteria are used. For example, the classical error exponent criterion characterizes the exponential rate at which the probability of error decays as $N$ and $n$ increase to infinity. In addition to assessing a particular classifier's performance, it is desirable to establish fundamental limits on the best achievable performance. %This is useful to assess tests and usually sheds light on possible way for improvement of a test.  

In many applications such as text classification, the number of training and test samples observed, $N$ and $n$, are much smaller than the size of alphabet $m$. This is the so-called \emph{sparse sample} problem. For example, suppose we want to decide, given two articles written by two different others, which author writes the third article. The number of words appearing in an article is much smaller than the English vocabulary, and the histogram of words is a sparse one \cite{zhaole01p5}. 

%One approach to study this sparse sample classification problem is the high-dimensional setting in which $N,n,m$ all tend to infinity.   
The high-dimensional setting, in which $N,n,m$ all tend to infinity and $m$ is much large than $N,n$, is a widely-used approach to analyze classifiers for the sparse sample problem. A widely-used performance criterion is asymptotic consistency: Given some dependence of $N,n$ on $m$, does the probability of error decay to zero as $m$ increases to infinity? A fundamental result with respect to this criterion was established in \cite{keltulwagvis10p1478}: Assuming that the distribution on all symbols in the alphabet is of order $1/m$, there exists an asymptotic consistent classifier if and only if $m=o(n^2)$. Note that the result is established only for the case $N=n$.

In most practical scenarios, the number of test samples available is smaller than the number of training samples. It is thus desirable to understand how $N$ and $n$ affects the  performance individually. We thus pose the following questions: \begin{enumerate}
\item How fast do $N$ and $n$ need to increase with $m$ in order to have an asymptotic consistent classifier? 
\item Does the probability of error depend on $N$ and $n$ in the same way? 
\item If the number of training samples is limited, can the  performance be improved by having more test samples? %, and vice versa? 
\end{enumerate}

The goal of this paper is to answer these questions by establishing achievability and converse results on best achieble probability of classification error. Our tool is the generalized error exponent analysis technique from \cite{huamey12}. In this prior work, the sparse sample \emph{goodness of fit} problem is investigated in which the number of test samples is much smaller than the size of alphabet. The classical error exponent was extended to this problem via a different scaling in large deviation analysis. 

In the classification problem, the classsical error exponent analysis has been applied to the case of fixed alphabet in \cite{ziv88p278} and \cite{gut89p401}. It was shown that in order for the probability of error to decay exponentially fast with respect to $n$, the number of training samples $N$ must grow at least linearly with $n$. However, in the sparse sample  problem, the classical error exponent concept is again not applicable, and thus a different scaling is needed. 

We identify the appropriate scaling in this paper, and thereby obtain a generalized error exponent to approximate 
the probability of error for large but sparse observations.This analysis yields new insights on the best achievable performance:
\begin{enumerate}
\item The numbers of training and test samples $N,n$ have different effects on the performance, made precise in \Theorem{achievability} and \Theorem{converse}.
\item The $\ell_2$-norm based classifier investigated in \cite{keltulwagvis10p1478}, which compares the $\ell_2$ distances from the empirical distribution of the test sequence to those of the two training sequences, is sub-optimal in that it has zero generalized error exponent,  while a weighted coincidence-based classifier proposed in this paper has a non-zero generalized error exponent. 
\end{enumerate}

{\bf \em Related work:} %In the case of a fixed alphabet, such results on fundamental limits have been established in \cite{ziv88p278} and \cite{gut89p401}: 
Two problems that are closely related to the sparse sample classification problem is the goodness of fit problem and the problem of testing whether two distributions are close. For the goodness of fit problem, achievability and converse results with respect to different criteria have been established in \cite{batfisforkumrubwhi01p442, pan08p4750, erm98p589, bar89p107, huamey12}. For the problem of testing the closeness of two distributions, achievability and converse results with respect to asymptotic consistency have been established in \cite{batforrubsmiwhi00p259, val08p383}.
Our converse result uses the concept of profile in \cite{orlsanjun04p1469}. The results in \cite{orlsanjun04p1469} have lead to algorithms for classification and closeness testing \cite{Acharya10, achjaforlpan11p47}.

\section{Notation and Model}
Consider the following classification problem: Two training sequences $\bfmX$ and $\bfmY$ are generated i.i.d.\ with marginal distributions $\pi$ and $\mu$, respectively. Each symbol takes value in $[m] \eqdef\{1,2,\ldots, m\}$. A test sequence $\bfmZ$ is observed. The sequence $\bfmZ$ is i.i.d.\ with marginal distribution $\pi$ under the null hypothesis $\Ho$ and with marginal $\mu$ under the alternative hypothesis $\Ha$. The three sequences $\bfmX, \bfmY, \bfmZ$ are independent.

Denote the set of probability distributions over $[m]$ by $\bmP$. 
The pair of unknown distributions $(\pi, \mu)$ belongs to the following set $\mP  \subseteq \bmP \times \bmP$,
\begin{equation}
\mP=\{(\pi,\mu): \|\mu-\pi\|_1 \geq  \epsy, \max_{j}\pi_j \leq \frac{\mub}{m}, \max_{j}\mu_j  \leq \frac{\mub}{m}\},\nonumber\\
\end{equation}
where $\mub$ is a large positive constant.
%\[\|\mu-\pi\|_1=2\sup\{|\mu(A)-\pi(A)|:A \subseteq [m]\}.\]
The definition of $\mP$ is essentially the same as the $\alpha$-large-alphabet source defined in \cite{keltulwagvis10p1478}, except that we allow the number of training and test samples to be different. While this assumption that all words are rare does not hold for English vocabulary, the insights and classifiers obtained for rare words will be used to improve the algorithms for the case when there are both frequent and rare words. 

The assumption that $\max_{j}\pi_j \leq \frac{\mub}{m}, \max_{j}\mu_j  \leq \frac{\mub}{m}$ indicates that we are interested in how the existence of a large number of rare symbols affects the performance, and is motivated by the English vocabulary. Extending the results to the case where there are both rare and non-rare symbols is a topic currently under investigation. 

%One may also question the use of $\ell_1$ metric. This is motivated by the form of the large deviation rate function: While the Kullback-Leibler divergence in the fixed alphabet case plays the role of large deviation rate function, it is the normalized $\ell_2$  metric $\sqrt{m}\|\mu-\pi\|_2$ that appears in the rate function for the sparse sample case. Under the assumption on $\pi$ and $\mu$, the $\ell_1$ metric is essentially the same as the normalized $\ell_2$ metric for fixed $\epsy$: $\|\mu-\pi\|_1 \leq \sqrt{m}\|\mu-\pi\|_2 \leq \sqrt{2\mub \|\mu-\pi\|_1}$. 

In the high-dimensional model, we consider a sequence of classification problems as described above, indexed by $m$. Thus $\bmP, N, n, p, q, \mP$ all depend on $m$. Moreover, $N,n$ increase to infinity as $m$ increases. 

A classifier $\bphi=\{\phi_m\}_{m \geq 1}$ is a sequence of binary-valued functions with $\phi_m: [m]^N \times [m]^N \times [m]^n \rightarrow \{0,1\}$. It decides in favor of $\Ha$ if $\phi_m=1$ and $\Ho$ otherwise. Use the notation
$\Prob_{(\mu,\pi,\nu)}(A)$ to denote the probability of the event $A$ when $\bfmX$, $\bfmY$ and $\bfmZ$ have marginal distributions $\mu,\pi,\nu$ respectively. The performance of a classifier $\phi$ is evaluated using the worst-case average probability of error given by 
\[P_e(\phi_m)\!=\!\sup_{(\pi,\mu)\in \mP}\![\half\Prob_{(\pi,\mu,\pi)}\{\phi_m\!=\!1\}+\half\Prob_{(\pi,\mu,\mu)}\{\phi_m\!=\!0\}].\]
It is said to be asymptotically consistent if 
\[\lim_{m \rightarrow \infty} P_e(\phi_m)=0.\]

\section{Asymptotic Consistency}

We begin with the asymptotic consistency result. 
\begin{theorem}\label{t:consistency}
There exists an asymptotically consistent classifier if and only if
\begin{equation}\label{e:consistencyonlyif}
m=o(\min\{N^2,  Nn\}).\nonumber
\end{equation}
\end{theorem}
\noindent 
\begin{proof}
The sparse sample case where $\max\{N,n\}=o(m)$ is a  corollary of the generalized error exponent analysis results given in \Theorem{achievability} and \Theorem{converse}. 

Now consider the case when $m=O(N)$. The only if direction is trivial. For the if direction, when $m=o(N)$, the distributions of $\bfmX$ and $\bfmY$ can be essentially be estimated with vanishing error since the number of types grows sub-exponentially in $n$ (See \cite[Lemma 3]{keltulwagvis10p1478}). When $m$ is linear in $N$, this problem can be transformed into a (harder) sparse sample problem with alphabet size $m b $ where $b=\lceil \sqrt{\min\{N, n\}} \rceil$: Associate each symbol in $[m]$ with $b$ symbols. Each observation is then randomly mapped to one of the associated symbols. A consistent classifier for the sparse sample  problem leads to a consistent classifier for the original problem. 
\end{proof}
We have a few remarks: 
\begin{enumerate}
\item For the case $N=n$, the conclusion of \Theorem{consistency} is consistent with the results in \cite[Theorem 3 and 4]{keltulwagvis10p1478}.Our proof technique is different. 
%\item We conjecture that \Assumption{Nngrowth}, only used in the proof of the achievability result, is unnecessary.  
\item The requirements on $N$ and $n$ for asymptotic consistency are different: The first requirement $m=o(N^2)$ needs to be satisfied regardless of  how many test samples are available. The second requirement is active only when $n=O(N)$. Therefore, as long as the number of test samples grows linearly with the training samples, further increasing the test samples will not improve the performance in terms of asymptotic consistency. 
\item On the other hand, increasing the number of \emph{training} samples will always increase the performance. The effect of increasing the training samples is different when $n=o(N)$ and $N=o(n)$. 
%\item Note that we have the assumption that $N=o(m)$. The number of test samples required must at least satisfy $\lim_{m \rightarrow \infty} n/(m/N)=\infty$. This is consistent with the well-known fact that it is impossible for the test to be asymptotically consistent if $n=O(1)$.
\end{enumerate}

\section{Generalized Error Exponent}

When $m$ is fixed, the following  error exponent criterion  has been used to evaluate a classifier $\bphi$:
\begin{equation}\label{e:defexponentchern}
\cJ(\phi)  \eqdef -\limsup_{n \rightarrow \infty} \frac{1}{n} \log(P_e(\phi_m)).
\end{equation}
This classical error exponent criterion is no longer applicable in the sparse sample case where 
\begin{assumption0}\label{as:Nngrowth}
$N=o(m), n=o(m)$.
\end{assumption0}
One should consider instead  the following generalization, defined with respect to the normalization $r(N,n,m)$:
\begin{equation}\label{e:defexponent}
J(\phi)  \eqdef -\limsup_{n \rightarrow \infty} \frac{1}{r(N,n,m)}\log(P_e(\phi_m)).
\end{equation}
The results in \Theorem{achievability} and \Theorem{converse} imply that the appropriate normalization is 
\[\vspace{-1mm}r(N,n,m)=\min\{N^2, Nn\}/m.\]
The generalized error exponent $J(\phi)$ could depend on how $N,n$ increase with $m$. Note that to have a consistent classifier, the necessary condition in \Theorem{consistency} must be satisfied, as summarized in the assumption below:
\begin{assumption0}\label{as:Nngrowth2} $m=o(\min\{N^2,  Nn\})$.
\end{assumption0} 
\noindent This is equivalent to $\lim_{m \rightarrow \infty} r(N,n,m)=\infty$.
%it is necessary that $m=o(\min\{N^2, Nn\})$; thus $r(N,n,m)$ tends to infinity. %Note that $m=o(n^2)$ is necessary, since otherwise no test is asymptotically consistent   \cite{pan08p4750}.

The following theorems demonstrate that the definition in \eqref{e:defexponent} is meaningful:
\begin{theorem}[Achievability]\label{t:achievability}
Suppose \Assumption{Nngrowth} and \Assumption{Nngrowth2} hold. Then there exists a classifier $\phi$ such that 
\[J(\phi)>0.\]
\end{theorem}
\begin{theorem}[Converse]\label{t:converse}
Suppose \Assumption{Nngrowth} holds. There exists a constant $\bar{J}$ such that for any classifier $\phi$,
\[-\log(P_e(\phi_m)) \leq r(N,n,m)\bar{J}.\]
\end{theorem}
These theorems imply that the best achievable probability of error decays approximately as $P_e=\exp\{-r(N,n,m)J\}$ for some $J>0$. Note that the probability of error changes exponetially with respect to $n$ only when $n=O(N)$. When $N=o(n)$, the probability of error is mainly determined by the number of training samples. This phenomenon is similar to the case with fixed $m$, for which results in \cite{ziv88p278} show that whether $n=O(N)$ holds determines whether the probability of error decreases exponentially in $n$.

%In the rest of the paper, we will first demonstrate that the well-known $\ell_2$-norm based classifier has a zero generalized error exponent in \Section{sec:l2normclassifier}; we will then propose a classifier that achieves non zero error exponent and prove \Theorem{achievability} in \Section{sec:weightclassifier}. The converse result is proved in \Section{sec:converse}.

\section{$\ell_2$-norm based classifier has a zero generalized error exponent}\label{sec:l2normclassifier}
Let $a^{z}_j$ be the number of times that $j$th symbol appears in $\bfmZ$. The notations $a^{x}$ and $a^{y}$ are defined similarly. 

The $\ell_2$-norm based classifier has the following test statistic:
\[F_n\eqdef \|\frac{1}{n}a^{z}-\frac{1}{N}a^{x}\|_2^2 -\|\frac{1}{n}a^{z}-\frac{1}{N}a^{y}\|_2^2.\]
The classifier is given by
\[\phi^F=\ind\{F_n \geq 0\}.\]
This classifier was shown in \cite{keltulwagvis10p1478} to be asymptotically consistent when $N=n$ and $m=o(N^2)$. 
We now show, however, this classifier has zero generalized error exponent: 
\begin{theorem}
Suppose  \Assumption{Nngrowth} and \Assumption{Nngrowth2} hold and $N=n$. Assume in addition that $m=o(n^2/\log(n)^2)$. Then 
\[J(\phi^F)=0.\]
\end{theorem}
The sub-optimality of $\phi^F$ is due to the following reason: For any $j$, a large variation of the value of $a^{y}_j$ causes a significant change in the value of the statistic $F_n$. 
Assume $m$ is even for simplicity of exposition. Let $u$ denote the uniform distribution on $[m]$. Let $q_j=(1+\epsy)/m$ for $j \leq m/2$ and $q_j=(1+\epsy)/m$ for $j >m/2$. Consider the case where under $\Ho$, the distribution is given by $(q,u,q)$. 

Considering the following event where one symbol appears many times:
\begin{equation}\label{e:constructAn}
C_n\eqdef \{a^{y}_1=\lfloor 4n/\sqrt{m}\rfloor\},
\end{equation}
we claim that this event is likely to cause a false alarm:
\[\Prob_{(q,u,q)} \{\phi^F=1|C_n\}=1-o(1).\]
On the other hand, the probability of $C_n$ decays slowly: 
\begin{equation}\Prob_{(q,u,q)}(C_n)=\exp\{-4(n/\sqrt{m})\log(m) (1+o(1))\}.
\end{equation}
Combining these two equality gives the lower-bound 
\[\begin{aligned}
\log(P_e(\phi^F)) &\geq \log\bigl(\half \Prob_{(q,u,q)}(C_n)\Prob_{(q,u,q)} \{\phi^F=1|C_n\}\bigr)\\
& =3 4\frac{n}{\sqrt{m}}\log(m)(1+o(1))
\end{aligned}
\]
Thus this error decays at most as $nm^{-\half}\log(m)$, slower than $n^2/m$. Consequently, $J(\phi^F)=0$.

\section{Proof of achievability: weighted coincidence-based classifier}\label{sec:weightclassifier}
A nonzero generalized error exponent is achieved by the following  weighted coincidence-based classifier, whose construction is inspired by the weighted coincidence-based test proposed in \cite{huamey12}. 
Define the test statistic $T_n$:
\begin{equation}
\begin{aligned}
T_n\!=\!\!\sum_{j}\bigl[\!&\frac{1}{N^2}\ind\{a^{x}_j=2, a^{z}_j=0\}+\frac{1}{n^2}\ind\{a^{x}_j=0, a^{z}_j=2\}\\
&-\frac{1}{nN}\ind\{a^{x}_j=1, a^{z}_j=1\}\!+\!\frac{1}{nN}\ind\{a^{y}_j=1, a^{z}_j=1\}\\
&-\frac{1}{n^2}\ind\{a^{y}_j=0, a^{z}_j=2\}\!-\!\frac{1}{N^2}\ind\{a^{y}_j=2, a^{z}_j=0\} \,\,\bigr].
\end{aligned}\nonumber
\end{equation}
The classifier is given by $\phi^T=\ind\{T_n \geq 0\}$. 

\Theorem{achievability} is proved by bounding $P_e(\phi^T)$ via Chernoff: 
\[\begin{aligned}
\log(\Prob_{(\pi,\mu,\pi)}\{\phi^T=1\}) & \leq \inf_{\theta}\Lambda_{(\pi,\mu,\pi)}(\theta).\\
\log(\Prob_{(\pi,\mu,\mu)}\{\phi^T=0\}) & \leq \inf_{\theta}\Lambda_{(\pi,\mu,\mu)}(\theta).
\end{aligned}\]
where $\Lambda_{(\pi,\mu,\nu)}(\theta)=\log \Expect_{(\pi,\mu,\nu)}[\exp(\theta K_n)]$ is the logarithmic moment generating function of $K_n$. The main step is to obtain an asymptotic approximation to $\Lambda_{(\pi,\mu,\nu)}(\theta)$, given in the following proposition:
\begin{proposition}\label{t:lmgfbound}
Let $\theta=\min\{N^2,nN\}\gamma$. For $\gamma=O(1)$,
\begin{equation}
\begin{aligned}
&\Lambda_{(\pi,\mu,\nu)}(\theta)\\
\leq &{\min\{N^2,nN\}}\bigl(\gamma [\sum_{j=1}^m (\half(\pi_j-\nu_j)^2-\half(\mu_j-\nu_j)^2)]\\
&\qquad \qquad\qquad+\gamma^2[\sum_{j=1}^m (\pi_j\nu_j+\mu_j\nu_j)+ \half(\pi_j^2+\mu_j^2)]\bigr)\\
&+O(\frac{\min\{N^2,nN\} \max\{N,n\}}{m^2})+O(1).
\end{aligned}\nonumber%\label{e:lmgfbound}
\end{equation}
%For $\theta=O(1)$, 
%\begin{equation}
%\begin{aligned}
%&\Lambda_{(\pi,\mu,\nu)}(\min\{N^2,nN\}\gamma)\\
%=&\sum_{j=1}^m \bigl[nN\pi_j\nu_j(e^{-\theta /(nN)}-1)+nN\mu_j\nu_j(e^{\theta /(nN)}-1)\\
%&\qquad+\half N^2\pi_j^2(e^{\theta /(N^2)}-1)+\half N^2\mu_j^2(e^{-\theta /(N^2)}-1)\bigr]\\
%&+O(\frac{Nn \max\{N,n\}}{m^2})+O(1).
%\end{aligned}\label{e:lmgf}
%\end{equation}
\end{proposition}
%Let $\theta=\min\{N^2,nN\}\gamma$, and consider $|\gamma|\leq 1$. It follows form Taylor series expansion that
%\[e^{\theta/(nN)}-1\leq \theta/nN+2\theta^2/(nN)^2\]
%Similar upper-bound holds for $e^{-\theta/(nN)}-1$, $e^{\theta/(N^2)}-1$, and $e^{-\theta/(N^2)}-1$. 
%Consequently, 
%\begin{equation}
%\begin{aligned}
%&\Lambda_{(\pi,\mu,\nu)}(\theta)\\
%\leq &\theta  [\sum_{j=1}^m (\half(\pi_j-\nu_j)^2-\half(\mu_j-\nu_j)^2)]\\
%&+2\theta^2[\sum_{j=1}^m (\pi_j\nu_j+\mu_j\nu_j)/(nN)+ \half(\pi_j^2+\mu_j^2)/N^2]\\
%&+O(\frac{Nn \max\{N,n\}}{m^2})+O(1)\\
%\leq &\frac{\min\{N^2,nN\}}{m}\bigl(\gamma [\sum_{j=1}^m (\half(\pi_j-\nu_j)^2-\half(\mu_j-\nu_j)^2)]\\
%&\qquad \qquad\qquad\quad+10\gamma^2 \mub^2\bigr)\\
%&+O(\frac{Nn \max\{N,n\}}{m^2})+O(1).
%\end{aligned}\label{e:lmgfbound}
%\end{equation}\
\Proposition{lmgfbound} is obtained using the Poisonnization technique: The distribution of the vector $a_j^x$ is the same as the conditional distribution of a vector of Poisson random variables whose expected values are given by $\lambda \pi$ for some constant $\lambda >0$, conditioned on the event that the sum of these random variables is equal to $N$. The main steps are similar to those used for results in \cite{huamey12}.

Applying \Proposition{lmgfbound} with the Chernoff bound for the cases $\nu=\pi$ and $\nu=\mu$, and using \Assumption{Nngrowth} and \Assumption{Nngrowth2}, and the facts $\pi_j,\mu_j \leq \mub/m$ and $\sum_{j=1}^m (\mu_j-\pi_j)^2 \geq \epsy^2/m$, we obtain
\begin{equation}
\begin{aligned}
\log(\Prob_{\pi,\mu,\pi}\{\phi^T=1\})&\leq -\frac{\epsy^4}{160\mub^2}\frac{\min\{N^2,nN\}}{m}(1+o(1)),\\
\log(\Prob_{\pi,\mu,\mu}\{\phi^T=0\})&\leq -\frac{\epsy^4}{160\mub^2}\frac{\min\{N^2,nN\}}{m}(1+o(1)).
\end{aligned}\nonumber
\end{equation}
Note that the approximation $o(1)$ is uniform over all $(\pi,\mu)\in\mP$. Therefore, \[J \geq \frac{\epsy^4}{160\mub^2}.\]

\section{Proof of converse}\label{sec:converse}
{\bf \em Step 1}: Establish the upper bound, 
\begin{eqnarray}\label{e:conversebound1}
-\log(P_e(\phi_m)) \leq \bar{J}_1N^2/m.
\end{eqnarray}
The main idea of the proof is to consider a event under which observations do not give any information regarding the hypotheses, and lower-bound the probability of such a event. 

We now make this precise. Define the event 
\[
\begin{aligned}
A=\{&\textrm{No symbol in $\bfmX$ appears more than once;}\\
&\textrm{no symbol in $\bfmY$ appears more than once.}\}
\end{aligned}\] 
Assume without loss of generality that $m$ is even. Define a collection of bi-uniform distributions as follows: Let $K_m$ denote the collection of all subsets of $[m]$ whose cardinality is $m/2$. For each set $\omega \in K_m$, define the distribution $\qk$ as
\begin{equation}
\qk_j=\left\{\begin{array}{l l}(1+\epsy)/m, & j \in \omega;\\ (1-\epsy)/m, & j \in [m]\setminus \omega.
\end{array}\right.\end{equation}
Note that $\|u-\qk\|_1=\epsy$, and $(u,\qk) \in \mP$ for all $\omega$.
%Then
%\[\mathcal{Q}=\{\qk: \omega \in K_m\}.\]
%Note that$(u,\qk) \in \mP$ for all $\omega$.

We will use the short-hand notation $\{(\bfmx,\bfmy,\bfmz)\}=\{(\bfmX,\bfmY,\bfmZ)=(\bfmx,\bfmy,\bfmz)\}$ throughout the paper. 

Our choice of the collection of distributions makes sure that the following result holds:
\begin{lemma}\label{t:equalprobA}
For any sequence $(\bfmx,\bfmy,\bfmz) \subseteq A$, % satisfying $\{(\bfmX,\bfmY,\bfmZ)=(\bfmx,\bfmy,\bfmz)\} \subseteq A$,
\[\frac{1}{|K_m|}\!\sum_{\omega \in K_m}\!\Pr_{(u,\qk,u)}(\bfmx,\bfmy,\bfmz)=\frac{1}{|K_m|}\!\sum_{\omega \in K_m}\!\Pr_{(\qk, u,u)}(\bfmx,\bfmy,\bfmz).\]
\end{lemma}
\begin{proof}[Proof sketch for \Lemma{equalprobA}]
For any sequence, let $\varphi_i$ denote the number of symbols appearing $i$ times. The vector $[\varphi_1, \varphi_2, \varphi_3, \ldots]$ is called the \emph{profile} of the sequence \cite{orlsanjun04p1469}. 

Because of the symmetry of the collection of distributions $\{\qk, \omega \in K_m\}$, the symmetry of the uniform distribution $u$, and the independence among $\bfmX,\bfmY,\bfmZ$, the value of $\frac{1}{|K_m|}\sum_{\omega \in K_m}\!\Pr_{(u,\qk,u)}(\bfmx,\bfmy,\bfmz)$ only depends on the profiles of $\bfmx$, $\bfmy$, and $\bfmz$. In the event $A$, the profiles of $\bfmx$ and $\bfmy$ are fixed, which then leads to the claim of the lemma. 
\end{proof}

\Lemma{equalprobA} implies that for any observation $(\bfmx,\bfmy,\bfmz) \in A$, it is impossible to tell whether it is more likely to come from the mixture on the left-hand side or the mixture on the right-hand side. Consequently,
%\begin{equation}
%\begin{aligned}
%&\frac{1}{|K_m|}\!\sum_{\omega }[\Pr_{(u,\qk,u)}(\{\phi_m\!=\!1\}\!\cap\! A)\!+\!\Pr_{(\qk,u,u)}(\{\phi_m\!=\!0\}\!\cap \!A)]\\
%=&\frac{1}{|K_m|}\!\sum_{\omega }[\Pr_{(u,\qk,u)}(\{\phi_m\!=\!1\}\!\cap\! A)\!+\!\Pr_{(u,\qk,u)}(\{\phi_m\!=\!0\}\!\cap\! A)]\\
%=&\frac{1}{|K_m|}\!\sum_{\omega }\Pr_{(u,\qk,u)}(A).
%\end{aligned}\nonumber
%\end{equation}
%A lower-bound on $\frac{1}{|K_m|}\sum_{\omega}\Pr_{(u,\qk,u)}(A)$ in \Lemma{lowerboundA} thus gives a lower-bound on the probability of error:
\begin{equation}
\begin{aligned}
&P_e(\phi_m)\\
\geq&\frac{1}{4|K_m|}\!\sum_{\omega } [\Prob_{(u,\qk,u)}\{\phi_m\!=\!1\}\!+\!\Prob_{(u,\qk,\qk)}\{\phi_m\!=\!0\}]\\
&\!+\! \frac{1}{4|K_m|}\!\sum_{\omega } [\Prob_{(\qk,u,\qk)}\{\phi_m\!=\!1\}\!+\!\Prob_{(\qk,u, u)}\{\phi_m\!=\!0\}]\\
\geq& \frac{1}{4|K_m|}\!\sum_{\omega }[\Pr_{(u,\qk,u)}\{\phi_m\!=\!1\}\!+\!\Pr_{(\qk,u,u)}\{\phi_m\!=\!0\})]\\
\geq& \frac{1}{4|K_m|}\!\sum_{\omega }[\Pr_{(u,\qk,u)}(\{\phi_m\!=\!1\}\!\cap\! A)\!+\!\Pr_{(\qk,u,u)}(\{\phi_m\!=\!0\}\!\cap \!A)]\\
=&\frac{1}{4|K_m|}\!\sum_{\omega }[\Pr_{(u,\qk,u)}(\{\phi_m\!=\!1\}\!\cap\! A)\!+\!\Pr_{(u,\qk,u)}(\{\phi_m\!=\!0\}\!\cap\! A)]\\
=& \frac{1}{4|K_m|}\!\sum_{\omega}\Pr_{(u,\qk,u)}(A).
\end{aligned}\label{e:PeboundbyPA}
\end{equation}
where the first inequality follows from the fact that the maximum is no smaller than the average, and the second last inequality follows from \Lemma{equalprobA}. The probability of the event $A$ can be lower-bounded. 
\begin{lemma}\label{t:lowerboundA}
%The value of  $\Pr_{(u,\qk,u)}(A)$ does not depend on $\omega$. Moreover, 
The following approximations holds uniformly for any $\omega$:
\[\log\bigl(\Pr_{(u,\qk,u)}(A)\bigr)=-(1+\half \epsy^2)\frac{N^2}{m}(1+o(1))+O(1).\]
\end{lemma}
\begin{proof}[Proof sketch]
It follows from a combinatorial argument that the probability that no symbol appears twice in $\bfmX$ when $\bfmX$ has marginal distribution $u$ is given by 
\[m (m-1) \ldots (m-N+1) (1/m)^N=\exp\{-\half\frac{N^2}{m}(1+o(1))\}. \]
Estimating the probability that no symbol appears twice in $\bfmY$ can be done similarly but is more involved.
\end{proof}

The claim \eqref{e:conversebound1} follows from applying \Lemma{lowerboundA} to \eqref{e:PeboundbyPA}, and picking a large enough  $\bar{J}$.

{\bf \em Step 2}: Establish the second upper-bound
\begin{eqnarray}\label{e:conversebound2}
-\log(P_e(\phi_m) \leq \bar{J}_2(Nn+n^2)/m.
\end{eqnarray}
We consider the following event:
\[
\begin{aligned}
B=\{&\textrm{No symbol in $\bfmZ$ appears more than once;}\\
&\textrm{no symbol in $\bfmZ$ has appeared in either $\bfmX$ or $\bfmY$}\}.
\end{aligned}\]
%\[B=\{\textrm{No symbol in $\bfmZ$ has appeared in either $\bfmX$ or $\bfmY$}\}.\]
When this event happens, it is impossible (in the worst-case setting) to infer which distribution the test sequence is more likely to be generated from. This is captured by the following lemma:
\begin{lemma}\label{t:equalprobB}
Consider any $\bfmx,\bfmy$. For any two sequences $\bfmz$ and $\bar{\bfmz}$ such that $(\bfmx,\bfmy,\bfmz) \subseteq B$ and $(\bfmx,\bfmy,\bar{\bfmz}) \subseteq B$, the following holds:
% satisfying $\{(\bfmX,\bfmY,\bfmZ)=(\bfmx,\bfmy,\bfmz)\} \subseteq A$,
\[\frac{1}{|K_m|}\!\sum_{\omega \in K_m}\!\Pr_{(u,\qk,u)}(\bfmx,\bfmy,\bfmz)=\frac{1}{|K_m|}\!\sum_{\omega \in K_m}\!\Pr_{u, \qk,u}(\bfmx,\bfmy,\bar{\bfmz}).\]
\[\frac{1}{|K_m|}\!\sum_{\omega \in K_m}\!\Pr_{(u,\qk,\qk)}(\bfmx,\bfmy,\bfmz)=\frac{1}{|K_m|}\!\sum_{\omega \in K_m}\!\Pr_{u, \qk,\qk}(\bfmx,\bfmy,\bar{\bfmz}).\]
\end{lemma}
\begin{proof}[Proof sketch for \Lemma{equalprobB}]
Since no symbols in $\bfmz$ have appeared in $\bfmx$ and $\bfmy$, due to the symmetry of the collection of distributions $\{\qk, \omega \in K_m\}$ and the symmetry of the uniform distribution $u$, for fixed $\bfmx$ and $\bfmy$, the value of $\frac{1}{|K_m|}\!\sum_{\omega \in K_m}\!\Pr_{(u,\qk,\qk)}(\bfmx,\bfmy,\bfmz)$ only depends on the profile of $\bfmz$. It follows from the definition of the event $B$ that the profile of $\bfmz$ is the same as the profile of $\bar{\bfmz}$. 
\end{proof}

The result  of \Lemma{equalprobB} can interpreted as follows: In the event $B$, observing $\bfmZ$ does not gives any information since under either hypothesis, each sequence $\bfmz$ appears with equal probability. % the same probability to appear. 
 
Consider any $\bfmx,\bfmy$. Let $D_{\bfmx,\bfmy}=\{\bfmz: (\bfmx,\bfmy,\bfmz) \in \{\phi_m=1\}\cap B\}$ and $D^c_{\bfmx,\bfmy}=\{\bfmz: (\bfmx,\bfmy,\bfmz) \in \{\phi_m=0\}\cap B\}$. \Lemma{equalprobB} implies that the probability of $\{\bfmX=\bfmx,\bfmY=\bfmy, \phi_m=1\}\cap B$ only depends on the size of $D_{\bfmx,\bfmy}$, rather than what sequences the set $D_{\bfmx,\bfmy}$ includes. Consequently, 
% \Let $B_{\bfmx,\bfmy}=\{(\bfmx,\bfmy,\bfmz) \in B\}$. 
We then have%It follows from \Lemma{equalprobB} that
\begin{equation}
\begin{aligned}
&\frac{1}{|K_m|}\!\sum_{\omega }\bigl[\Pr_{(u,\qk,u)}(\{\bfmX=\bfmx,\bfmY=\bfmy, \phi_m=1\}\cap B)\\
&\qquad\qquad+\Pr_{(u,\qk,\qk)}(\{\bfmX=\bfmx,\bfmY=\bfmy, \phi_m=0\}\cap B)\bigr]\\
=&\bigl[\frac{1}{|K_m|}\!\sum_{\omega }\Pr_{(u,\qk,u)}(\{\bfmX\!=\!\bfmx,\bfmY\!=\!\bfmy\}\!\cap\! B)\bigr]\frac{|D_{\bfmx,\bfmy}|}{D_{\bfmx,\bfmy}\!+\!D^c_{\bfmx,\bfmy}}\\
&+\!\bigl[\frac{1}{|K_m|}\!\sum_{\omega }\!\Pr_{(u,\qk,\qk)}\!(\{\bfmX\!=\!\bfmx,\bfmY\!=\!\bfmy\}\!\cap\! B)\bigr]\frac{|D^c_{\bfmx,\bfmy}|}{D_{\bfmx,\bfmy}\!+\!D^c_{\bfmx,\bfmy}}\\
\geq & \frac{1}{|K_m|}\min\bigl\{\sum_{\omega }\Pr_{(u,\qk,u)}(\{\bfmX\!=\!\bfmx,\bfmY\!=\!\bfmy\}\!\cap\! B),\\
&\qquad\qquad\quad\sum_{\omega }\Pr_{(u,\qk,\qk)}\!(\{\bfmX\!=\!\bfmx,\bfmY\!=\!\bfmy\}\!\cap\! B)\bigr\},
\end{aligned}\label{e:PeboundbyPB}
\end{equation}
where the inequality follows from lower-bounding the probability of $\{\bfmX\!=\!\bfmx,\bfmY\!=\!\bfmy\}\!\cap\! B$ under $(u,\qk,u)$ and $(u,\qk,\qk)$ by the minimum of these two. 

%equality follows from \Lemma{equalprobB} which implies that the probability of $\phi_m=1$ for a given $\bfmx,\bfmy$ and under event $B$
%\begin{equation}
%%\frac{1}{|K_m|}\!\sum_{\omega }[\Pr_{(u,\qk,u)}\!(\{\phi_m\!=\!1\}\!\cap \!B)=\frac{|{}
%\begin{aligned}
%&\frac{1}{|K_m|}\!\sum_{\omega }[\Pr_{(u,\qk,u)}\!(\{\phi_m\!=\!1\}\!\cap \!B)\!+\!\Pr_{(u, \qk,\qk)}\!(\{\phi_m\!=\!0\}\!\cap \!B)]\\
%=&\frac{1}{|K_m|}\!\sum_{\omega }[\Pr_{(u,\qk,u)}(\{\phi_m\!=\!1\}\!\cap \!B)\!+\!\Pr_{(u,\qk,u)}(\{\phi_m\!=\!0\}\!\cap\! B)]\\
%=&\frac{1}{|K_m|}\!\sum_{\omega }\Pr_{(u,\qk,u)}(B).\nonumber
%\end{aligned}
%\end{equation}
%Consequently,
%\begin{equation}
%\begin{aligned}
%&\frac{1}{|K_m|}\!\sum_{\omega }[\Pr_{(u,\qk,u)}\!(\{\phi_m\!=\!1\}\!\cap \!B)\!+\!\Pr_{(u, \qk,\qk)}\!(\{\phi_m\!=\!0\}\!\cap \!B)]\\
%=&\frac{1}{|K_m|}\!\sum_{\omega }[\Pr_{(u,\qk,u)}(\{\phi_m\!=\!1\}\!\cap \!B)\!+\!\Pr_{(u,\qk,u)}(\{\phi_m\!=\!0\}\!\cap\! B)]\\
%=&\frac{1}{|K_m|}\!\sum_{\omega }\Pr_{(u,\qk,u)}(B).\nonumber
%\end{aligned}
%\end{equation}
%Therefore,\begin{equation}
%\begin{aligned}
%&P_e(\phi_m)\\
%\geq& \frac{1}{2|K_m|}\!\sum_{\omega }[\Pr_{(u,\qk,u)}\!(\!\{\phi_m\!=\!1\}\!\cap\! B)\!+\!\!\Pr_{(u, \qk,\qk)}\!(\!\{\phi_m\!=\!0\}\!\cap\! B)]\\
%=& \frac{1}{2|K_m|}\!\sum_{\omega}\Pr_{(u,\qk,u)}(B).
%\end{aligned}\label{e:PeboundbyPB}
%\end{equation}
\begin{lemma}\label{t:lowerboundB}
Let $\bar{J}_2=5$. Then the following bounds hold uniformly over all $\omega, \bfmx,\bfmy$:
\[\log\!\bigl[\!\frac{\Pr_{(u,\qk,u)}\!(\!\{\bfmX\!=\!\bfmx,\bfmY\!=\!\bfmy\}\!\cap\! B)}{\Pr_{(u,\qk,u)}\!\{\bfmX\!=\!\bfmx,\bfmY\!=\!\bfmy\}}\!\bigr]\!\geq\! \bar{J}_2\frac{N\!n\!+\!n^2}{m}(1+o(1)).\]
\[\log\!\bigl[\!\frac{\Pr_{(u,\qk,\qk)}\!(\!\{\bfmX\!=\!\bfmx,\bfmY\!=\!\bfmy\}\!\cap\! B)}{\Pr_{(u,\qk,\qk)}\!\{\bfmX\!=\!\bfmx,\bfmY\!=\!\bfmy\}}\!\bigr]\!\geq\! \bar{J}_2\frac{N\!n\!+\!n^2}{m}(1+o(1)).\]
%\[\begin{aligned}
%&\log\big(\frac{\Pr_{(u,\qk,u)}(\{\bfmX\!=\!\bfmx,\bfmY\!=\!\bfmy\}\!\cap\! B)}{\Pr_{(u,\qk,u)}\{\bfmX\!=\!\bfmx,\bfmY\!=\!\bfmy\}})\\
%\geq& -2\frac{Nn+0.25(1+\epsy^2)n^2}{m}(1+o(1)).
%\end{aligned}
%\]
%\[\begin{aligned}
%&\log\big(\frac{\Pr_{(u,\qk,\qk)}(\{\bfmX\!=\!\bfmx,\bfmY\!=\!\bfmy\}\!\cap\! B)}{\Pr_{(u,\qk,\qk)}\{\bfmX\!=\!\bfmx,\bfmY\!=\!\bfmy\}})\\
%\geq& -2\frac{Nn+0.25(1+\epsy^2)n^2}{m}(1+o(1)).
%\end{aligned}
%\]
%The value of  $\Pr_{(u,\qk,u)}(B)$ does not depend on $\omega$. Moreover, the following bound holds:
%\[\log\bigl(\Pr_{(u,\qk,u)}(B)\bigr)\geq -2\frac{Nn+0.25(1+\epsy^2)n^2}{m}(1+o(1)).\]
\end{lemma}
\noindent The proof is similar to that of \Lemma{lowerboundA}.

Note that the average probability of error is equal to the summation of the left-hand side of \eqref{e:PeboundbyPB} over all possible $(\bfmx,\bfmy)$. Applying \Lemma{lowerboundB} to lower-bound the right-hand side of \eqref{e:PeboundbyPB} leads to the claim. 

We now combine \eqref{e:conversebound1} and \eqref{e:conversebound2}. It is straightforward to verify that 
\[\min\{N^2, Nn+n^2\} \leq \min\{N^2, 2Nn\}.\]
Taking $\bar{J}=\max\{\bar{J}_1,2\bar{J}_2\}$ leads to the claim of the theorem.

%\vspace{2cm}

\section{Conclusions and Future Work}
%We have studied binary classification when the size of the underlying alphabet $m$ is larger than the number of training samples $N$ and test samples $n$. We show that there is an asymptotically consistent classifier if and only if $m=o(\min\{N^2,Nn\})$. Moreover, we characterize the rate of convergence using generalized error exponent: The best achievable probability of error is $P_e=\exp\{- J \min\{N^2,Nn\} (1+o(1)\}$. The results shed light on the different roles played by the training samples and test samples. We propose a weighted coincidence-based classifier that achieves $J>0$, and also show that the known $\ell_2$-norm based classifier has zero generalized error exponent. 

%The above results are established for the case $N,n=o(m)$, and also that all symbols are rare, i.e., the probability of observing any symbols is on the same order $m^{-1}$. We plan to address the following directions in future work:
We have investigated the binary classification problem with sparse samples using generalized error exponent concept, and established fundamental performance limits. We have proposed a classifier that performs better than the $\ell_2$-norm based classifier. Future directions include:
\begin{enumerate}
\item Investigate classification algorithms that are applicable when there are both rare and frequent symbols.
\item The generalized error exponent analysis could be  applicable to the problem of testing closeness of  distributions. 
%Given the connection to the problem of testing closeness of two distributions, it is desirable to extend the generalized error exponent analysis to this problem. 
\end{enumerate}

%\begin{spacing}{0.4}
{
\bibliographystyle{IEEEtran}
\bibliography{IEEEfull,DayuStringFull,Dayu_bibtex}
}
%\end{spacing}
\end{document}